\pdfoutput=1
\newif\ifFull
\Fulltrue
\ifFull
\documentclass[11pt]{llncs}
\usepackage[margin=1in]{geometry}
\let\realsubsection\subsection
\else
\documentclass{llncs}  
\let\realsubsection\subsection
\renewcommand{\subsection}[1]{\paragraph{\bf #1.}}
\fi

\usepackage{graphicx}
\usepackage{cite}

\begin{document}
\ifFull
\pagestyle{plain}
\fi
\renewenvironment{proof}{\noindent{\bf Proof:}}{\hspace*{\fill}\rule{6pt}{6pt}\bigskip}

\title{Combinatorial Pair Testing:\\ 
       Distinguishing Workers from Slackers}

\author{David Eppstein \and Michael T. Goodrich \and Daniel S.~Hirschberg}

\institute{
Dept. of Computer Science, University of California, Irvine, CA 92697 USA
}

\date{}

\maketitle

\begin{abstract}
We formalize a problem we call 
\emph{combinatorial pair testing} (CPT), which has applications
to the identification of uncooperative or unproductive participants 
in pair programming, massively distributed computing,
and crowdsourcing environments.
We give efficient adaptive and nonadaptive CPT algorithms and we show
that our methods use an optimal number of testing rounds to within
constant factors.
We also provide an empirical evaluation of some of our methods.
\end{abstract}

\section{Introduction}
\emph{Pair programming}~\cite{WilKes-03} 
is a software development paradigm where
programmers are teamed in pairs and write software together using a single
workstation.
This paradigm is said to produce fewer software bugs
and shorter programs than when programmers work alone~\cite{WilKesCun-IS-00}.
Consequently, it is often used to teach software design in
introductory programming courses~\cite{NagWilFer-SIGCSE-03}, 
including courses at the authors' institution~\cite{Jacobson:2008}, 
the University of California, Irvine.
This design paradigm presents an additional 
challenge, however, for evaluative purposes.
Namely, if programmers are always
working in pairs, how can a manager or instructor evaluate the performance of
programmers as individuals?

For instance, suppose 100 students enroll in an introductory programming course, among whom 80 are conscientious and 20 are lazy. We will call the conscientious students
\emph{workers} and the lazy ones \emph{slackers}.
In order to assign final grades to these students, the instructor would like to distinguish the workers from
the slackers, but whenever she pairs a worker and a slacker
on a project, the worker will do the assignment
individually and the project will be completed successfully in spite of the
slacker's laziness. 
Based on their performances, the instructor can only detect
slackers when two slackers are paired together.
Therefore, it would be useful for her to have systematic and effective
strategies for pairing the students in order to distinguish workers from slackers.

Motivated by this evaluation problem,
we are interested in this paper in
the design of efficient algorithms for generating testing schemes that can distinguish
workers from slackers.
We formulate such problems in a general framework, which we call
\emph{combinatorial pair testing} (CPT), and we consider a number of different
assessment settings, such as whether
all tests must be specified in advance
or whether tests may be determined adaptively.
This approach allows us to focus on natural performance
characteristics of such problems and
provides a general framework that unifies other 
diagnosis problems under the CPT heading.

\subsection{Combinatorial Pair Testing}
Suppose we are given a set $X$ of $n$ individuals, $\epsilon n$ of whom
are \emph{slackers} and $(1-\epsilon)n$ of whom are \emph{workers},
where $\epsilon$ may or may not be known in advance.
A \emph{pairwise test} is a function $T(x,y)$ that takes as its arguments two members
$x$ and $y$ of $X$, 
and produces as output a Boolean value, the result of a test 
performed for $x$ and $y$ based solely on the 
worker/slacker status of $x$ and $y$.
Naturally,
although this framework allows for $T$ to be any Boolean function,
some Boolean functions will be more interesting
than others. 
In this paper,
we are particularly interested in the following type of test:
\begin{itemize}
\item
\emph{Performance-based testing:} In a performance-based test,
we pair two individuals, $x$ and $y$,
and evaluate their output performance as a team.
Thus, if both $x$ and $y$ are slackers, then
$T(x,y)=\textbf{false}$, indicating that 
the two slackers, $x$ and $y$, have been paired 
together and didn't complete the assigned project.
If, on the other hand, $x$, $y$, or both, are workers, then
$T(x,y)=\textbf{true}$, indicating that 
the project was completed.
\end{itemize}
Performance-based testing is symmetric, 
so $T(x,y)=T(y,x)$, and, indeed, this test is 
equivalent to a Boolean OR of $x$ and $y$, 
where a slacker corresponds to a 0 and a worker corresponds to a 1.
Moreover, by 
De~Morgan's laws, any CPT algorithm that uses OR for $T(x,y)$ can 
be easily modified to produce a CPT algorithm that uses AND for $T(x,y)$.

In \emph{combinatorial pair testing} (CPT), only
pairwise tests are allowed. The tests are organized in a sequence of \emph{rounds}, 
in which each member of $X$ may be tested at most once, so up to $\lfloor n/2\rfloor$ pairwise tests can be performed in a single round. 
The choices made by CPT algorithms
can be determined adaptively or non-adaptively
and may be based on decisions that are either deterministic or randomized.
In some cases we will also require some prior knowledge of the relative numbers
of slackers and workers; for instance, using only performance-based tests,
it is not possible to distinguish the case of there being
only one slacker in $X$ from that of there being none.
Moreover, the efficiency of a given testing scheme may depend on
assumptions about the number of slackers.

Because our intended applications may involve sensitive information about individual misbehavior, we may also desire CPT algorithms to have additional security or privacy 
guarantees.
For instance, we may want our algorithms to be implementable in a way that
allows an instructor to outsource the evaluation of the tests without
revealing the input data~\cite{Atallah:2008}.
Such an approach is common in privacy-preserving computations (e.g., see
\cite{Yao1986}).

One additional security condition that we study in this paper, which appears to
be novel, is that of a detection algorithm
that is \emph{participant oblivious}.
A detection algorithm is participant oblivious if an individual
cannot detect whether he has been identified by the evaluator as a worker or
slacker based only on the pairings to which he has been assigned (without knowing the status of his or her partners or the outcome of their tests).
A nonadaptive algorithm must be participant oblivious, but we show
that some adaptive algorithms can also be participant oblivious.
The advantage of a participant-oblivious algorithm is that it allows the
evaluator to impose penalties to slackers or rewards to workers after the completion of the tests without tipping off a participant during
the testing process that the evaluator might already know his or her status.

\subsection{Prior Related Work}
Combinatorial pair testing
is related to \emph{combinatorial group testing}~\cite{du2000}. 
In combinatorial group testing, we are given a set, $S$, of $n$ items,
at most $d$ of which are ``defective.''
A test consists of selecting a subset, $T$, and determining whether $T$ contains any defective items.
Thus, combinatorial pair testing with performance-based testing is a 
restricted type of combinatorial group test in which every subset is a pair.
There are many known results and applications for algorithmic problems
in combinatorial group testing
(e.g., see~\cite{du2000,Eppstein:2006,gat-iidf-05}), 
but we are not aware 
of any results for the case where every subset must be a pair and in which tests
are issued in groups of $O(n)$ independent tests.
The closest previous analysis is by Hwang~\cite{Hwang:2000}, 
who analyzes random size-$k$ tests that are issued independently 
(that is, not in groups).
Instead, insisting that every test to be a pair and that the pairs
are issued in groups, as is required
in combinatorial pair testing, goes against a standard 
approach in combinatorial group testing, according to which one performs tests to limit 
the defective items
to a subset of size at most $O(d\log n)$ and then tests each such item
individually.

Combinatorial pair testing is also a generalization of
\emph{processor fault diagnosis}. In this problem,
we are given a set of $n$ processors, each of which can be 
either faulty or good. One processor can check another, but the
result of this check can only be trusted if the
processor doing the testing is good.
Often, in fact, one assumes that faulty processors deliberately misidentify
the ones they are testing~\cite{Blecher1983107,Pelc:1998}.
Beigel {\it et al.}~\cite{bhk-95,Beigel:1989,Beigel:1993} 
show that if the number of faulty
processors is sufficiently far below $n/2$, then
$O(n)$ tests can be organized into
a sequence of $O(1)$ parallel testing rounds, where each processor tests at most one
other in each round, so as to identify all faulty processors.
Thus, processor fault diagnosis forms a type of
combinatorial pair testing problem where the tests are
based on queries and, in the case when faulty processors deliberately
misidentify the ones they are testing, the Boolean function that determines the outcome of a test is the exclusive-or function.

In addition, combinatorial pair testing can be applied to 
cheater detection in 
\emph{massively distributed computations}~\cite{Goodrich2008199}, such as 
SETI@home and distributed.net.
These systems break very large computations into 
independent tasks,
which are then sent out to be executed to participants of the system (typically by using the idle time of individual personal computers).
The problem is that some participants cheat: instead of performing the
requested tasks, they rig their computers to 
return false or partial results, often merely for the sake of appearing on a
leader board of top participants.
To deal with this problem, these distributed systems often will send out the
same task to two participants at the same time, and if they both return the
same answer, then the output is accepted and the participants are labeled as
being honest (e.g., see~\cite{dg-acns-05}).
One challenge is that when two answers don't agree, the system doesn't immediately know
which participant(s) cheated.
The problem of identifying all the honest participants (and, hence, all the cheaters)
in a distributed computing environment can be formulated using the approach of this paper, and solved, using combinatorial
pair testing with performance-based tests based on the AND function.
Previous work on cheater-detection in distributed computations does not take
this approach, however, and is instead based on ad hoc 
solutions or reductions to processor fault diagnosis
(e.g., see~\cite{dg-acns-05,Du:2004,Goodrich2008199}).

Along these same lines,
combinatorial pair testing also has applications to 
\emph{crowdsourcing}, where complex, independent tasks, such as labeling 
images, is farmed out to a large set of individuals to perform.
One challenge in this case is that the group of individuals contains 
both ``experts,''
who are competent and diligent with their work, and ``spammers,''
whose performance is no better than a random oracle~\cite{NIPS2012_0328}.
Combinatorial pair testing can be applied in this context to weed
out the spammers, much in the same way
as it applies to cheater detection for massively distributed computations.

\subsection{Our Results}
Given a set, $X$, of $n$ individuals such that $\epsilon n$ of them 
are slackers,
we formalize the combinatorial pair testing (CPT) problem,
and we present and analyze several efficient CPT algorithms for identifying
the slackers in $X$.
For the adaptive case, we give an algorithm that uses 
$O(1/\epsilon)$ testing rounds, and we show this to be optimal to within constant factors.
Moreover, we show that our algorithm is participant oblivious
and we extend our algorithm to work in $O(1/\epsilon)$
testing rounds even if we don't know the value of $\epsilon$
in advance.
We also give both deterministic and randomized nonadaptive CPT algorithms, 
and we show that the performance of these algorithms is optimal to
within constant factors.
For example, our randomized nonadaptive CPT algorithm uses 
$O((1/\epsilon)\log n)$ testing rounds and succeeds in identifying all slackers
with high probability.
Our analysis of this algorithm
is based on an extension to the coupon collectors problem,
which we call the coupon packet collectors problem.
In addition, we give an empirical study of our randomized CPT algorithm that
provides experimental bounds for the number of tests needed to identify various
percentages of the slackers in $X$.

\section{Adaptive Algorithms}
In this section, we describe an adaptive participant-oblivious algorithm for identifying all the slackers in a performance-based testing problem.

\subsection{The Two-Phase Algorithm}
Assume that we know
there are $\epsilon n$ slackers.

In phase one, we perform the following computation:
\begin{itemize}
\item
\emph{Phase One:}
We group the individuals into $\lfloor \epsilon n/2\rfloor$ ``bins'' of size 
at most $\lceil 2/\epsilon\rceil$ each.
We then do $\lceil 2/\epsilon\rceil$ 
``round-robin'' rounds of testing to compare all pairs of
items in the same bin as each other, across all bins in parallel. 
\end{itemize}
This completes phase one, and gives us the following.

\begin{lemma}
\label{lem:two-slackers}
After phase one completes, we will have
identified all the slackers in each bin that has at least 2 slackers.
\end{lemma}
\begin{proof}
If a bin contains 0 or 1 slackers, then each pairing of two individuals in that bin will contain a worker. Thus, every test for that 
bin has the same outcome (true).
If, on the other hand, a bin contains 2 or more slackers, then each slacker 
in that bin
will eventually be paired with another slacker; hence, we discover each
slacker in that bin.
\end{proof}

More importantly, we also have the following.

\begin{lemma}
After phase one completes, we will have
identified at least $\lceil \epsilon n/2\rceil$
slackers.
\end{lemma}
\begin{proof}
By the previous lemma, a slacker can go undiscovered only if he is the sole
slacker assigned to a given bin.
Since there are $\lfloor \epsilon n/2\rfloor$ bins, then,
by a generalized pigeonhole argument, there has to be 
at least $\epsilon n - \lfloor \epsilon n/2\rfloor = \lceil \epsilon n/2\rceil$
slackers that are assigned to bins that each contain at least two slackers.
\end{proof}

Given that we now have identified at least $\lceil \epsilon n/2\rceil$
slackers, in phase two we perform the following computation.
\begin{itemize}
\item
\emph{Phase Two:}
We choose $\lceil \epsilon n/2\rceil$ known slackers and assign one of them to each bin randomly.
We assign the remaining individuals to bins, while keeping the bins to be of size
at most $\lceil 2/\epsilon\rceil$.
Moreover, we choose these assignments uniformly at random, 
subject to the rule that each
bin contains a slacker and that no two individuals who were paired in round one are assigned to the same bin as each other.
We then do $\lceil 2/\epsilon\rceil$ 
``round-robin'' rounds of testing to compare all pairs of
items in the same bin as each other, across all bins in parallel. 
\end{itemize}
This completes phase two. 

From the perspective of any individual, their bin
assignment 
is done at random, with every bin being equally likely, and the people they are paired with are equally likely to come from any other bin from phase one.
Moreover, the only nonadaptive step is the assignment of known slackers to
bins in phase two, which is done via a random permutation, similar to how
elements not known to be slackers are assigned.
Thus, so long as individuals in our group do not collude, this algorithm is
participant oblivious.
Note, in addition, that any bin that now contains a previously undiscovered
slacker, will necessarily contain at least two slackers.
Thus, by Lemma~\ref{lem:two-slackers},
we will discover this (and all other) remaining slackers in phase two.

\begin{theorem}
Given a set, $X$, of $n$ workers and slackers, such that $\epsilon n$ of the
individuals in $X$ are known to be slackers, 
we can identify all the slackers in $X$
in $O(1/\epsilon)$ rounds of disjoint pairwise tests, in a participant-oblivious adaptive fashion. 
\end{theorem}

This bound is optimal, to within constant factors, as the following
theorem establishes.

\begin{theorem}
Given a set, $X$, of $n$ workers and slackers, such that $\epsilon n$ of the
individuals in $X$ are slackers, then identifying all the slackers in $X$ (either deterministically in the worst case or randomly with success probability $\ge 1/2$)
requires at least
$\Omega(1/\epsilon)$ rounds of disjoint pairwise tests.
\end{theorem}

\begin{proof}
We consider the randomized case first, and we assume a randomized input distribution in which all permutations of workers and slackers are equally likely. Let $x$ be a random variable whose value is one of the slackers in the input, chosen uniformly at random among the slackers. In the first $1/(2\epsilon)$ rounds of testing, at most $n/2$ of the members of $X$ may become identified.
 In any given round of testing in which $x$ has not already been identified as a slacker, at most $\epsilon n -1$ of the unidentified members of $X$ can be paired with (identified or unidentified) slackers other than $x$, and $x$ is equally likely to be any one of the $\ge n/2$ unidentified members,  so the probability that $x$ becomes identified by being paired with a slacker is at most $(\epsilon n-1)/(n/2)<\epsilon/2$. By the union bound, after $1/(4\epsilon)=\Omega(1/\epsilon)$ rounds, $x$ will remain unidentified with probability greater than $1/2$, so the probability that all slackers are identified is less than $1/2$.

Since this randomized input distribution fools even a randomized algorithm with probability at least $1/2$, after $\Omega(1/\epsilon)$ rounds, it follows that for every deterministic algorithm there exists an input in this distribution that is certain to fool the algorithm with the same number of rounds. 
\end{proof}

\subsection{Estimating Epsilon}
Suppose now that there are
$\epsilon n$ slackers, but we do not know the value of $\epsilon$.
Instead, let us assume we
have an estimate, $\epsilon'$, and our goal is to 
use $O(1/\epsilon')$ rounds,
and either find all $\epsilon n$ slackers, with
$\epsilon'\le 2\epsilon$,
or determine that $\epsilon'>\epsilon$.

Consider again the above two-phase algorithm, but now 
assume that it is calibrated for $\epsilon'$ instead of $\epsilon$.
One possible outcome of phase one, is that we discover
at least $\lceil \epsilon' n/2\rceil$ slackers, which then allows us 
to discover all the slackers in phase two.
In this case,
\[
\epsilon n \ge \epsilon' n/2,
\]
hence,
$\epsilon' \le 2\epsilon$.

Alternatively,
phase one may discover fewer than 
$\lceil \epsilon' n/2\rceil$ slackers. 
Since a bin that appears to hold no slackers can hold at most one,
this implies that 
\[
\epsilon<\epsilon'/2 + \epsilon'/2=\epsilon'.
\]
Thus, our two-phase algorithm achieves our goal.

We can therefore now use our two-phase algorithm 
in an iterative fashion. We start with $\epsilon'=1/2$, and use the two-phase
algorithm with this estimate for $\epsilon$.
If we discover all the slackers, then we are done.
Otherwise, we determine that $\epsilon<\epsilon'$. 
In this case, we set
$\epsilon'\leftarrow \epsilon'/2$ and we repeat the process with this
estimate.
Eventually, we will reach a point where we discover all the slackers, with
$\epsilon'\le 2\epsilon$.
Moreover, since the previous iteration, if there is one, would have 
failed, we also know that $\epsilon<2\epsilon'$, that is, 
$\epsilon'>\epsilon/2$.
The number of testing rounds is therefore proportional to
\[
2+4+8+\cdots + 1/\epsilon' \le 
2+4+8+\cdots + 2/\epsilon \le 4/\epsilon.
\]
Therefore, even without knowing the value of $\epsilon$,
the number of rounds is $O(1/\epsilon)$,
which implies the following.

\begin{theorem}
Given a set, $X$, of $n$ workers and slackers, such that $\epsilon n$ of the
individuals in $X$ are slackers, 
we can identify all the slackers in $X$
in $O(1/\epsilon)$ rounds of $O(n)$ pairwise tests
per round, in a participant-oblivious adaptive algorithm, 
without knowing $\epsilon$ in advance.
\end{theorem}

In
\ifFull
Section~\ref{sec:many}, we explore optimizations to
\else
the full version of this paper, we explore optimizations to
\fi
the constant factors in
the above bounds, in adaptive CPT algorithms
for the case when $\delta=1-\epsilon\le 1/2$, that is, when at least half
of the individuals are slackers.
Such instances of the combinatorial pair testing problem arise naturally
in massively distributed and crowdsourcing applications, for example,
where the roles of slackers and workers are reversed and the testing
function, $T$, is Boolean AND instead of OR.

\section{Nonadaptive Pair Testing}
In this section, we study nonadaptive algorithms for combinatorial 
pair testing, to identify $\epsilon n$ slackers in a group of $n$
individuals. In this case, if we assume that we 
do not know the value of $\epsilon$, then the only valid algorithm is the
trivial brute-force algorithm that compares every pair of individuals, since
a nonadaptive algorithm must specify all its tests in advance and it is
possible that $\epsilon=2/n$.
Therefore, we assume that we know in advance that there are $\epsilon n$ slackers.

\subsection{Deterministic Nonadaptive Pair Testing}
Unfortunately,
nonadaptive deterministic pair testing is not very interesting,
because it requires a linear number of rounds.
The argument is simple: suppose a deterministic nonadaptive pair testing algorithm could use at most $(1-\epsilon) n/2$
rounds. Then, in the graph of pairs that are tested by the algorithm,
each vertex would have at most $(1-\epsilon) n/2$ neighbors. An adversary could choose
one edge of the graph, make one of its two endpoints a slacker and the
other endpoint a worker, set all neighbors of these two vertices to be
workers, and fill out the rest of the graph arbitrarily to fit whatever
number of slackers and workers is desired.
From the set of tests that are performed, there is no way to distinguish
which of the two endpoints of the chosen edge is the slacker and which
is the worker.
Therefore, there must be at least $\Omega((1-\epsilon)n)$ rounds in a
deterministic nonadaptive CPT algorithm, which, for any fixed $\epsilon<1$,
is asymptotically not any better than the brute-force algorithm that tests
every pair.

This bound can be achieved as an upper bound, as well, using an algorithm
that pairs each individual, $x$, with at least $(1-\epsilon)n+1$ other distinct 
individuals, using $O((1-\epsilon)n)$ rounds. For this algorithm, at least one of the individuals
paired with each such $x$ must be a slacker.

\subsection{Randomized Nonadaptive Pair Testing}
Despite the nonexistence of efficient deterministic nonadaptive pair testing algorithms, there is a simple randomized algorithm 
for nonadaptive randomized testing, which
succeeds with high probability using many fewer tests than the deterministic
nonadaptive solution.
In particular, let us repeatedly choose a random matching of all the
members of the set, $X$, for some value, $k$, 
number of rounds.
Each matching corresponds to a round of testing. 
For instance, for $k=(c/\epsilon)\log n$, for 
a sufficiently large constant, $c\ge 1$,
then this scheme uses $O((1/\epsilon)\log n)$ rounds and
$O((n/\epsilon)\log n)$ tests in total.

\subsubsection{Relation to the coupon collector's problem.}
The expected performance 
of the nonadaptive randomized algorithm described above can be analyzed 
precisely using a variant of the classical \emph{coupon collector's problem}.

In the coupon collector's problem, a collector wishes to collect a set of $n$ trading cards, by randomly acquiring one card at a time,
and the problem is to calculate the number of steps that are required until, with high probability, all cards have been collected. 
Now consider a slight variation, 
which we call the \emph{coupon packet collector's problem}:
instead of buying one card at a time, the collector buys the cards in
packets of $m$ cards~\cite{Stadje1990}.
Each packet of trading cards is guaranteed to have no
duplicates, and is uniformly random among all $m$-card samples of the whole
set of cards. How does this affect the total time required for the
collector? If $m$ is much smaller than $n$, the difference between this problem and the standard coupon collector's problem is very small: a random sample of $m$ cards,
each independently and uniformly randomly sampled, is very likely to be
duplicate-free. But if $m$ is a constant fraction of $n$,  then the avoidance of
duplicates in each packet is very likely to cause the number of packets that
the collector needs to collect to be smaller by a constant fraction than the
number that a one-at-a-time collector would need. But what is the fraction?

In the coupon packet collector's problem, the probability that a card
remains uncollected after $k$ rounds is $(1-m/n)^k$. So, after
$k$ rounds, by the linearity of expectation, the expected number
of uncollected cards is $n(1-m/n)^k$.
Thus, for $k=(1+\alpha)\log_{1/(1-m/n)} n$ rounds, 
the expected number of uncollected cards is $1/n^\alpha$;
hence, 
by Markov's inequality, 
with very high probability,
$1-1/n^{\alpha}$, 
all the cards are collected.

In the pair testing problem, observe that a
slacker's status is identified whenever the slacker is paired with another
slacker, and a student's status is identified whenever that student is
paired with a known slacker. If we allow these identifications to be made
retroactively ({\em i.e.}, once we find a known slacker we use that identity to
confirm as workers all the other students the slacker has already been
paired with) then there is a very simple criterion for whether
we have identified everybody: we have done so if and only if all students have been
paired at least once with a slacker. More weakly, we have identified all
slackers whenever the slackers have all been paired with another slacker in
some round of testing. Suppose that there are $m$ slackers and $n$ total
students. In each round, exactly $m$ students will be paired with slackers, so
it is very much like the coupon packet collector's problem, where the
trading cards in a packet correspond to the students that are paired with
slackers. There is a small complication, however: in the pair testing
problem the sets of students that are identified are not quite uniformly
random over all $m$-element subsets of students. In particular, the slackers
are slightly less likely to be paired with other slackers than the workers, because there
are fewer other slackers for them to be paired with.

To be precise, in the case that there are an even number of students,
 a slacker has probability exactly $(n-m)/(n-1)$ of remaining
unidentified after one round, because there are $n-1$ students the slacker could be paired with, each of which is equally likely, and $n-m$ of which (the workers) fail to identify the slacker.
The probability that a specific student is identified in any one round is independent of the same probability for the same student in a different round, 
so after
\[
k = 
(1+\alpha)\log_{\frac{n-1}{n-m}} m
\]
rounds, the probability that an individual slacker remains unidentified is 
$1/m^{1+\alpha}$. 
Similarly, a worker has probability exactly $(n-m-1)/(n-1)$ of not having been paired with a slacker after one round, and probability 
$1/(n-m)^{1+\alpha}$ of never having been paired with a slacker after
\[
k = 
(1+\alpha)\log_{\frac{n-1}{n-m-1}} (n-m)
\]
rounds.
Different students have probabilities that are not independent of each other, 
but by linearity of expectation
after
\[
k = 
(1+\alpha)\max\left\{\log_{\frac{n-1}{n-m}} m,\log_{\frac{n-1}{n-m-1}} (n-m))\right\}
\]
rounds the expected number of students who have not been paired with 
a slacker is $\min\{1/m^{\alpha},1/(n-m)^{\alpha}\}$, 
so by Markov's inequality, 
with high probability all students will be identified. 
In the case that there are an odd number of students, 
there are $n$ alternatives for each student in each round rather than $n-1$, 
so the number of rounds needed is instead
\[
k = 
(1+\alpha)(\log_{\frac{n}{n-m+1}} m+\log_{\frac{n}{n-m}} (n-m)).
\]
In either case, for $m=\epsilon n$ slackers, if we 
extend the above two bounds so that the number of rounds is increased to
\[
k = (1+\alpha)\log_{1/(1-\epsilon)} n,
\]
then the expected number of 
unclassified students is $1/n^{\alpha}$.
Thus, by Markov's inequality, there are no unclassified students
with high probability, $1-1/n^{\alpha}$.
Choosing $\alpha\ge 1$ to be a fixed constant, and using the inequality,
$x<-\ln (1-x)$, for $0<x<1$, we get the following result.

\begin{theorem}
Given a set, $X$, of $n$ individuals, such that $\epsilon n\ge 2$ of 
them are slackers and the rest are workers, we can distinguish the workers
and slackers using $O((1/\epsilon)\log n)$ rounds of 
random performance-based tests, with $O(n)$ tests per round, with high
probability, $1-1/n^c$, in a nonadaptive fashion, for any fixed constant
$c\ge 1$.
\end{theorem}

\begin{figure}[p]
\begin{center}
\includegraphics[width=.5\textwidth, angle=270]{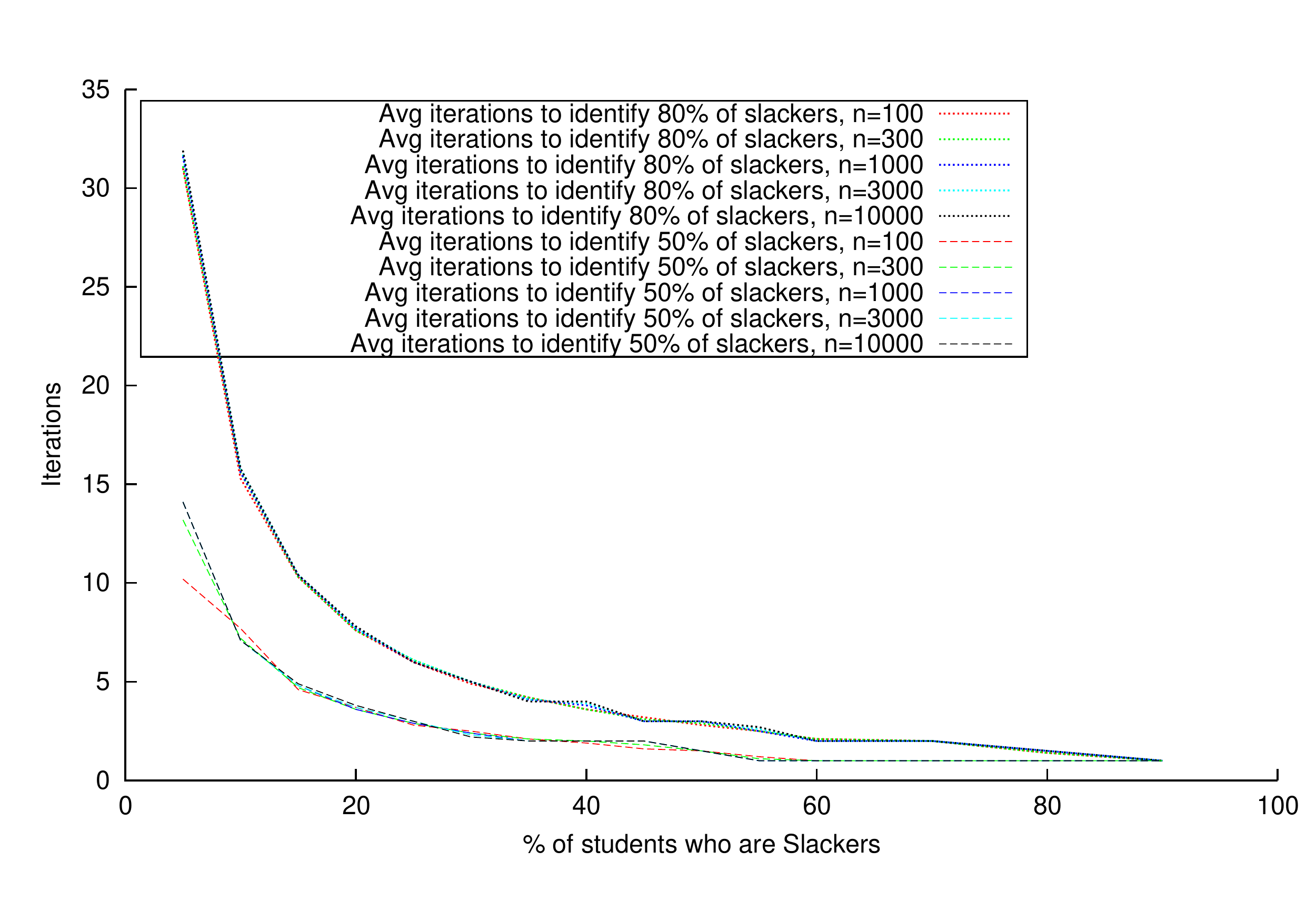} \\
\includegraphics[width=.5\textwidth, angle=270]{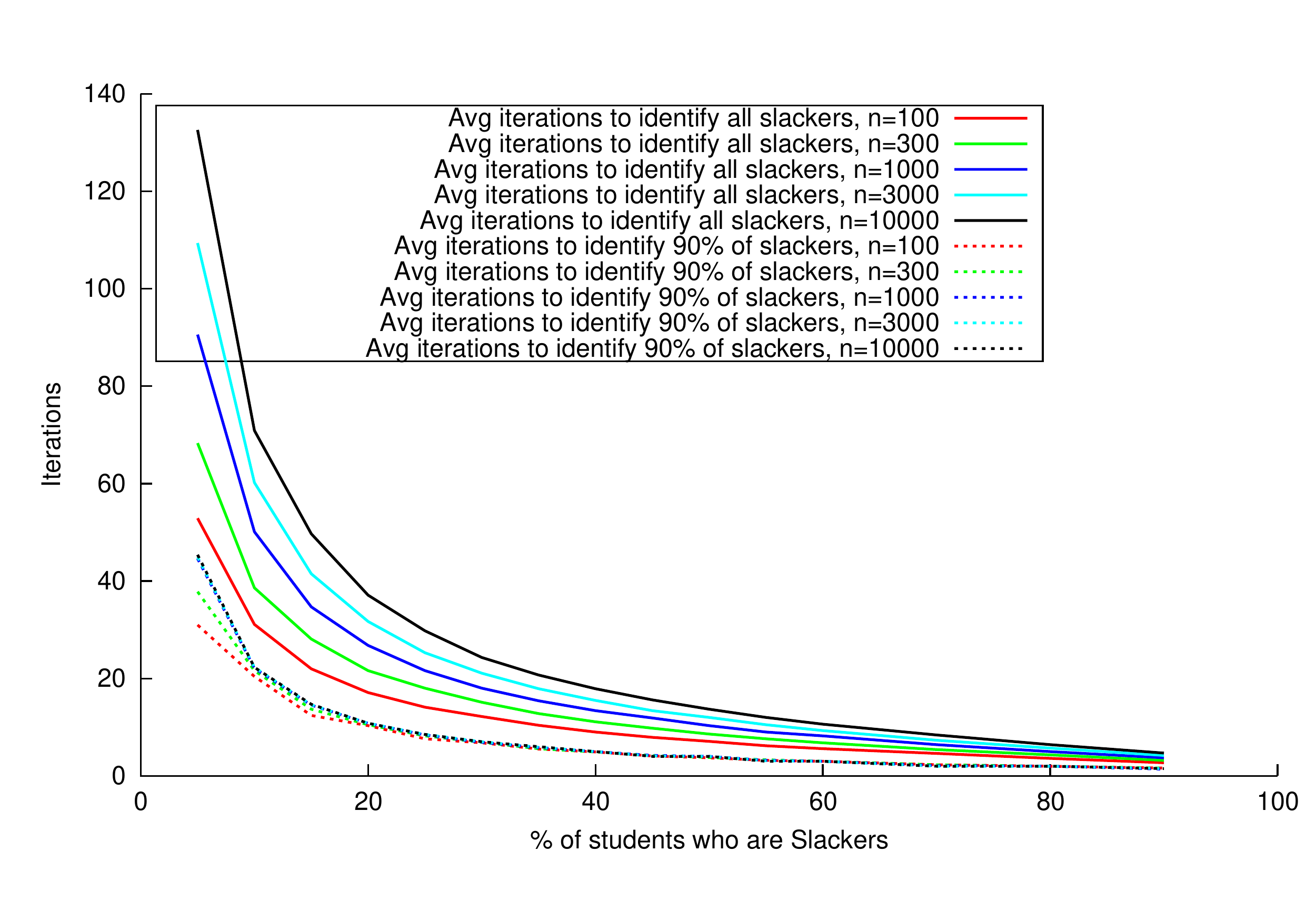} 
\caption{\label{fig-results} Results for number of random tests needed
to identify various percentages of slackers, for various values of the set
size, $n$, and slacker percentage, $\epsilon$.}
\end{center}
\end{figure}

In a nonadaptive randomized strategy, the most information is gathered by
randomly matching of the members of $X$ and testing each matched
pair.
Thus, for any slacker, $s$, the probability $s$ is not paired with 
another slacker is at least $(1-\epsilon)$.
So, after $k$ independent rounds of testing, the probability $s$ has not
been discovered to be a slacker is at least $(1-\epsilon)^k$, which
we can bound as
\[
(1-\epsilon)^k \ge \left( \frac{1-\epsilon}{e}\right)^{\epsilon k} ,
\]
by an inequality due to Niculescu and Vernescu~\cite{niculescu2004two}.
Thus, we have the following.

\begin{theorem}
For $2/n\le \epsilon\le 1/2$, we require $\Omega((1/\epsilon)\log n)$ rounds
of testing for each slacker to be identified, with probability at least
$1-1/n$, in a nonadaptive randomized testing scheme for a set of $n$
members having $\epsilon n$ slackers.
\end{theorem}

Therefore, the above analysis is tight to within constant factors.

\subsection{Experimental Results}
To get a better handle on the expected number of tests needed
to identify various percentages of slackers, we performed an experimental
study of the above nonadaptive randomized CPT algorithm.
We performed tests for values of $n$ ranging from 100 to 10000, with 
percentage of slackers ranging from 5\% to 90\%.
We then performed tests to determine the average number of tests required 
in order to identify 50\%, 80\%, 90\%, and 100\% of the slackers.
We show the results in Figure~\ref{fig-results}.

\ifFull
\section{Improved Bounds for Testing with Many Slackers}
\label{sec:many}
In this section, we consider the case when the
number of workers, $\delta n=(1-\epsilon)n$, is relatively
small, that is, when $\delta\le 1/2$.
We show that we can design a
non-trivial adaptive testing 
scheme that has a guaranteed small number of rounds of tests.

In general, if there are $\delta n$ workers then, after one round of tests,
at most $\delta n$ slackers will be paired with workers and, therefore, the other
at least $n - 2 \delta n$ slackers will have been identified.  
In the next rounds,
these known slackers can be paired with unknown students, thereby identifying them
(some of whom may be identified as slackers and others as workers).

\subsection{Two rounds}
If there are very few workers, say $\delta \leq 1/4$, then in one round of tests
at most $n/4$ slackers will be paired with workers and, therefore, the other
(at least $n/2$) slackers will be identified.  In a second round, we can pair
$n/2$ known slackers with the other $n/2$ students and thereby be able to identify
all other students.

\subsection{Three rounds}
If $1/4 < \delta \leq 1/3$ then at least $n/3$ slackers will be
identified in the first round and at most two more rounds of tests suffice to
identify all other students.  But we can do better.

If $1/4 < \delta \leq 5/14$, then at least $2n/7$ slackers will be
identified in the first round.
In this case, we can
conduct a second round of tests in which $2n/7$ known slackers are paired with unknown
students.
There are at least $9n/14$ slackers altogether
and so there will be at least $5n/14$ other slackers.
There are only $n/2$ pairs of which $2n/7$ are populated with known slackers,
so there are $3n/14$ pairs not populated with known slackers.
Therefore, after placing $3n/14$ of the other slackers in unoccupied pairs,
the remaining $n/7$ other slackers will be paired with slackers (known or unknown).

So, the second round of tests identifies at least $n/7$ slackers and,
since there had been $2n/7$ known slackers in separate pairs during this second round,
at least $n/7$ other students will also have been identified.
Therefore, after the second round of tests, there will have been identified
a total of at least $3n/7$ slackers and at least $n/7$ other students,
say $x$ other identified slackers and $n/7 -x$ identified workers.

If $x\geq n/14$ then at least $n/2$ slackers will have been identified and
a third round of tests suffices to identify all other students.
Otherwise, place the $3n/7+x$ known slackers in separate pairs and fill the remaining
$n/14-x$ pairs with $n/7-2x$ identified workers.
Again, a third round of tests suffices to identify all remaining students.

\subsection{Four rounds}
Similarly, if $5/14 < \delta \leq 19/46$ then at least $4n/23$ slackers will be
identified in the first round.
Following an analagous procedure and analysis, a second round of tests
identifies at least $2n/13$ slackers and at least $2n/13$ other students.
Proceed with a third round of tests in which all identified slackers are placed
in separate pairs and, to the extent possible, all identified workers are paired
with each other.  Using a similar analysis, at least an additional $3n/23$
slackers and at least $3n/23$ other students will be identified.
Finally, a fourth round of tests using the same placement strategy will suffice
to identify all remaining students.

\subsection{Five rounds}
If $\delta \leq 1/2$ then another approach enables all slackers to be identified
using at most five rounds of tests.
Partition the $n$ students into $n/4$ groups of four students each.
In three rounds of testing, each member of each group can have been paired with each
of the other members of that member's group.
If a group contains two or more slackers then the results of these three rounds of
tests will have enabled all members of that group to be identified.
If a group contains zero or one slacker then no members of the group will have been
identified.

If there are $\delta n$ workers, then there are $\epsilon n$ slackers,
where $\epsilon=1-\delta \geq 1/2$.

\begin{lemma}
There are at most $x = (1-\epsilon)n/3$ groups whose members will
not have been identified and therefore at most $x$ remaining unidentified slackers.
\end{lemma}
\begin{proof}
The lemma can be proven by counting the slackers.
Each of the $x$ unidentified groups has at most one slacker
and the other $(n/4)-x$ groups can each have at most four slackers.
The total number of slackers is $\epsilon n \leq x + 4((n/4)-x)$,
which directly proves the lemma.
\end{proof}

So, after three rounds of testing,
at most $x$ slackers will not have been identified and therefore
at least $\epsilon n - x = (4\epsilon-1)n/3 \geq n/3$ slackers
will have been identified.

Conduct a fourth round of testing in which each of $n/3$ identified slackers is
placed in a separate pair, thereby identifying $n/3$ of the other students.
A fifth round of testing will identify all remaining unknown students.

We have provided algorithms enabling identification of
all the slackers in 2, 3, 4, or 5 rounds, depending on the value of $\delta\le 1/2$.
We leave open the problem of establishing whether these algorithms are optimal.

\fi

\section{Conclusion}
In this paper, we have given efficient algorithms for solving combinatorial
pair testing problems, along with lower bounds showing that our algorithms
are optimal to within constant factors.
All of our algorithms assume we are using performance-based tests.
Therefore,
one possible direction for future work would be to explore CPT algorithms
and applications for other kinds of tests (other than the exclusive-or
tests used in 
processor fault 
diagnosis~\cite{bhk-95,Beigel:1989,Beigel:1993,Blecher1983107,Pelc:1998}). Another direction would be to enlarge the size of tested groups beyond two and explore the effect of different group sizes on the numbers of rounds needed for testing.

\realsubsection*{Acknowledgments} 
This research was supported in part by
the National Science Foundation under grants 1011840, 1217322, 
and 1228639, and by the Office of
Naval Research under MURI grant N00014-08-1-1015.

{\raggedright 
\bibliographystyle{abbrv} 
\bibliography{pairs} 

\begin{thebibliography}{10}

\bibitem{Atallah:2008}
M.~J. Atallah, K.~B. Frikken, M.~Blanton, and Y.~Cho.
\newblock {Private combinatorial group testing}.
\newblock In {\em ACM Symp on Information, Computer and Communications Security
  (ASIACCS)}, pages 312{--}320, 2008.

\bibitem{bhk-95}
R.~Beigel, W.~Hurwood, and N.~Kahale.
\newblock {Fault diagnosis in a flash}.
\newblock In {\em Proc. IEEE Foundations of Computer Science (FOCS)}, pages
  571{--}580, October 1995.

\bibitem{Beigel:1989}
R.~Beigel, S.~R. Kosaraju, and G.~F. Sullican.
\newblock {Locating faults in a constant number of parallel testing rounds}.
\newblock In {\em ACM Symp. on Parallel Algorithms and Architectures (SPAA)},
  pages 189{--}198, 1989.

\bibitem{Beigel:1993}
R.~Beigel, G.~Margulis, and D.~A. Spielman.
\newblock {Fault diagnosis in a small constant number of parallel testing
  rounds}.
\newblock In {\em ACM Symp. on Parallel Algorithms and Architectures (SPAA)},
  pages 21{--}29, 1993.

\bibitem{Blecher1983107}
P.~M. Blecher.
\newblock {On a logical problem}.
\newblock {\em Discrete Mathematics}, 43(1):107{--}110, 1983.

\bibitem{du2000}
D.-Z. Du and F.~Hwang.
\newblock {\em {Combinatorial Group Testing and Its Applications}}.
\newblock Series on Applied Mathematics. World Scientific, 2000.

\bibitem{dg-acns-05}
W.~Du and M.~T. Goodrich.
\newblock {Searching for high-value rare events with uncheatable grid
  computing}.
\newblock In J.~Ioannidis, A.~Keromytis, and M.~Yung, editors, {\em Applied
  Cryptography and Network Security (ACNS)}, volume 3531 of {\em LNCS}, pages
  122{--}137. Springer, 2005.

\bibitem{Du:2004}
W.~Du, J.~Jia, M.~Mangal, and M.~Murugesan.
\newblock {Uncheatable grid computing}.
\newblock In {\em 24th Int. Conf. on Distributed Computing Systems (ICDCS)},
  pages 4{--}11, 2004.

\bibitem{Eppstein:2006}
D.~Eppstein, M.~T. Goodrich, and D.~S. Hirschberg.
\newblock {Improved combinatorial group testing algorithms for real-world
  problem sizes}.
\newblock {\em SIAM J. Comput.}, 36(5):1360{--}1375, 2006.

\bibitem{Goodrich2008199}
M.~T. Goodrich.
\newblock {Pipelined algorithms to detect cheating in long-term grid
  computations}.
\newblock {\em Theoretical Computer Science}, 408(2/3):199{--}207, 2008.

\bibitem{gat-iidf-05}
M.~T. Goodrich, M.~J. Atallah, and R.~Tamassia.
\newblock {Indexing information for data forensics}.
\newblock In J.~Ioannidis, A.~Keromytis, and M.~Yung, editors, {\em Applied
  Cryptography and Network Security (ACNS)}, volume 3531 of {\em LNCS}, pages
  206{--}221. Springer, 2005.

\bibitem{Hwang:2000}
F.~K. Hwang.
\newblock Random {$k$}-set pool designs with distinct columns.
\newblock {\em Probab. Eng. Inf. Sci.}, 14(1):49--56, Jan. 2000.

\bibitem{Jacobson:2008}
N.~Jacobson and S.~K. Schaefer.
\newblock Pair programming in {CS1}: overcoming objections to its adoption.
\newblock {\em SIGCSE Bull.}, 40(2):93--96, June 2008.

\bibitem{NIPS2012_0328}
Q.~Liu, J.~Peng, and A.~Ihler.
\newblock {Variational inference for crowdsourcing}.
\newblock In P.~Bartlett, F.~C.~N. Pereira, C.~J.~C. Burges, L.~Bottou, and
  K.~Q. Weinberger, editors, {\em Advances in Neural Information Processing
  Systems (NIPS)}, pages 701{--}709, 2012.

\bibitem{NagWilFer-SIGCSE-03}
N.~Nagappan, L.~Williams, M.~Ferzli, E.~Wiebe, K.~Yang, C.~Miller, and
  S.~Balik.
\newblock {Improving the CS1 experience with pair programming}.
\newblock In {\em Proc. 34th SIGCSE Technical Symp. on Computer Science
  Education (SIGCSE '03)}, volume 35.1 of {\em SIGCSE Bulletin}, pages
  359{--}362, 2003.

\bibitem{niculescu2004two}
C.~P. Niculescu and A.~Vernescu.
\newblock {A two-sided estimate of $e^x-(1+ x/n)^n$}.
\newblock {\em Journal of Inequalities in Pure and Applied Mathematics}, 5(3),
  2004.

\bibitem{Pelc:1998}
A.~Pelc and E.~Upfal.
\newblock {Reliable fault diagnosis with few tests}.
\newblock {\em Comb. Probab. Comput.}, 7(3):323{--}333, 1998.

\bibitem{Stadje1990}
W.~Stadje.
\newblock {The collector's problem with group drawings}.
\newblock {\em Advances in Applied Probability}, 22(4):866{--}882, 1990.

\bibitem{WilKes-03}
L.~Williams and R.~R. Kessler.
\newblock {\em {Pair Programming Illuminated}}.
\newblock Addison-Wesley, 2003.

\bibitem{WilKesCun-IS-00}
L.~Williams, R.~R. Kessler, W.~Cunningham, and R.~Jeffries.
\newblock {Strengthening the case for pair programming}.
\newblock {\em IEEE Software}, 17(4):19{--}25, 2000.

\bibitem{Yao1986}
A.~C. Yao.
\newblock {How to generate and exchange secrets}.
\newblock In {\em Proceedings of the 27th Annual Symposium on Foundations of
  Computer Science}, pages 162{--}167, Washington, DC, USA, 1986. IEEE Computer
  Society.

\end{thebibliography}
}

\clearpage

\end{document}